\documentclass[a4paper,11pt,]{article}
\usepackage{graphicx}
\usepackage{tabularx}
 \usepackage{amssymb}
 \usepackage{amsmath}
\usepackage{multirow}
\usepackage{supertabular}

\newtheorem{THM}{Theorem}[section]
\newtheorem{LMA}[THM]{Lemma}

\global \count10 = 0

\global \count11 = 1

\global \count12 = 1

\global \count13 = 1

\newcommand{\com}[1]{\ifnum\count13<1 #1 \fi}

\def\squarebox#1{\hbox to #1{\hfill\vbox to #1{\vfill}}}
\def\qed{\hspace*{\fill}%
        \vbox{\hrule\hbox{\vrule\squarebox{.667em}\vrule}\hrule}\smallskip}
\newenvironment{proof}{\begin{trivlist}
\item[\hspace{\labelsep}{\em\noindent Proof.~}]}{\qed\end{trivlist}}

\def\squarebox#1{\hbox to #1{\hfill\vbox to #1{\vfill}}}
\def\qed{\hspace*{\fill}%
        \vbox{\hrule\hbox{\vrule\squarebox{.667em}\vrule}\hrule}\smallskip}

\begin{document}

\title{
An optimal algorithm for 2-bounded delay buffer management with lookahead
}

\author{
Koji M. Kobayashi
}

\date{}

\maketitle

\pagestyle{plain}
\thispagestyle{plain}

\begin{abstract}
	\ifnum \count10 > 0
	%
	Kesselman
	QoS
	%
	%
	%
	$d(p) - r(p) + 1 \leq s$
	$s$-bounded instance
	$s \geq 2$
	%
	%
	%
	%
	B{\"o}hm
	lookahead
	%
	%
	
	%
	%
	
	%
	\fi
	\ifnum \count11 > 0
	%
	%
	The bounded delay buffer management problem, 
	which was proposed by Kesselman et~al.\ (STOC 2001 and SIAM Journal on Computing 33(3), 2004), 
	is an online problem focusing on buffer management of a switch supporting Quality of Service (QoS). 
	The problem definition is as follows: 
	Packets arrive to a buffer over time 
	and each packet is specified by the {\em release time}, {\em deadline} and {\em value}. 
	An algorithm can transmit at most one packet from the buffer at each integer time 
	and can gain its value as the {\em profit} 
	if transmitting a packet by its deadline after its release time. 
	The objective of this problem is to maximize the gained profit. 
	We say that an instance of the problem is $s$-bounded 
	if for any packet, 
	an algorithm has at most $s$ chances to transmit it. 
	For any $s \geq 2$, 
	Hajek (CISS 2001) showed that the competitive ratio of any deterministic algorithm is at least $(1 + \sqrt{5})/2 \approx 1.619$. 
	It is conjectured that there exists an algorithm whose competitive ratio matching this lower bound for any $s$. 
	However, 
	it has not been shown yet. 
	Then, 
	when $s = 2$, 
	B{\"o}hm et al.~(ISAAC 2016) introduced the {\em lookahead} ability to an online algorithm, that is the algorithm can gain information about future arriving packets, 
	and showed that the algorithm achieves the competitive ratio of $(-1 + \sqrt{13})/2 \approx 1.303$. 
	Also, 
	they showed that the competitive ratio of any deterministic algorithm is at least $(1 + \sqrt{17})/4 \approx 1.281$. 
	In this paper, for the 2-bounded model with lookahead, 
	we design an algorithm with a matching competitive ratio of $(1 + \sqrt{17})/4$. 
	%
	\fi
\end{abstract}


\section{Introduction} \label{sec:intro}
\ifnum \count10 > 0
%
%
Aiello
Quality of Service (QoS)
%
%
%
%
Kesselman
{\em 
%
%
%
%
%
$A$
%
%
$d(p) - r(p) + 1 \leq s$
{\em $s$-
%
%
$s \geq 2$
%
%
%

%
%
%
%
%
%
%
%
B{\"o}hm
%

%
\fi
\ifnum \count11 > 0
%
%
The online buffer management problem proposed by Aiello et~al.~\cite{WA00} formulates the management of buffers to store arriving packets in a network switch with Quality of Service (QoS) support as an online problem. 
%
%
This problem has received much attention among online problems and has been studied for the last fifteen years, 
which leads to developing various variants of this problem 
(see comprehensive surveys \cite{MG10,SN15}). 
Kesselman et~al.~\cite{AKB01} proposed the {\em bounded delay buffer management} problem 
as one of the variants, 
whose definition is as follows: 
Packets arrive to a buffer over time. 
A packet $p$ is specified by the {\em release time} $r(p)$, {\em value} $v(p)$ and {\em deadline} $d(p)$. 
An algorithm is allowed to transfer at most one packet at each integer time. 
If the algorithm transmits a packet between its release time and deadline, 
it can gain its value as the {\em profit}. 
The objective of this problem is to maximize the gained profit. 
The performance of an online algorithm for this problem is evaluated using {\em competitive analysis}~\cite{AB98,DS85}. 
If for any problem instance, 
the profit of an optimal offline algorithm $OPT$ is at most $c$ times that of an online algorithm $A$, 
then we say that the competitive ratio of $A$ is at most $c$. 
We call a problem instance the {\em $s$-bounded instance} 
(or {\em $s$-bounded delay buffer management} problem) 
in which for any packet $p$, 
$d(p) - r(p) + 1 \leq s$. 
For any $s \geq 2$, 
Hajek~\cite{Haj01} showed that 
the competitive ratio of any deterministic algorithm is at least $(1 + \sqrt{5})/2 \approx 1.619$. 
Also, 
it is conjectured that for any $s \geq 2$, 
there exists a deterministic algorithm with a competitive ratio of $(1 + \sqrt{5})/2$ (see, e.g. \cite{MG10}), which has not been proved yet. 
There is much research among online problems to reduce the competitive ratio of an online algorithm for the original problems by adding extra abilities to the algorithm. 
One of the major methods is called the {\em lookahead} ability, 
with which an online algorithm can obtain information about arriving packets in the near future. 
This ability is introduced to various online problems: 
The bin packing problem~\cite{EG96}, 
the paging problem~\cite{SA97,DB98}, 
the list update problem~\cite{SA98}, 
the scheduling problem~\cite{RM98}
and so on. 
Then, 
B{\"o}hm et~al.~\cite{MB16} introduced the lookahead ability to the bounded delay buffer management problem, 
that is, 
they gave an online algorithm for this problem an ability to obtain the information about future arriving packets 
and analyzed its performance. 
\fi
\ifnum \count10 > 0
%
%
\noindent
\textbf{Previous Results and Our Results.}~~~
B{\"o}hm
lookahead
%
%
%

%
%

%
\fi
\ifnum \count11 > 0
%
%
\noindent
\textbf{Previous Results and Our Results.}~~~
B{\"o}hm et~al.~\cite{MB16} studied the 2-bounded bounded delay buffer management problem with lookahead. 
They designed a deterministic algorithm whose competitive ratio is at most $( - 1 + \sqrt{13})/2 \approx 1.303$. 
Also, 
they proved that the competitive ratio of any deterministic algorithm is at least $( 1 + \sqrt{17})/4 \approx 1.281$. 
In this paper, 
we showed an optimal online algorithm for this problem, that is, 
its competitive ratio is exactly $(1 + \sqrt{17})/4$. 
\fi
\ifnum \count10 > 0
%
%
\noindent
\textbf{
lookahed
$s \geq 2$
Hajek~\cite{Haj01}
%
%
$s = \infty$
Englert
%
%
$s = 2$~\cite{AKB01}, 3~\cite{BCC04,CCF06}, 4~\cite{MB16}
%
$s \geq 5$
%
%
{\em memoryless}
%
%
{\em agreeable
$r(p) < r(p')$
$d(p) \leq d(p')$
%
%
Hajek~\cite{Haj01}
Li
%
%
{\em $s$-uniform delay}
%
%
$(1 + \sqrt{5})/2$~\cite{LSS05,JLS12}%
$1.377$~\cite{CJST07}
%

%
%
%
$s$
oblivious adversary
$e/(e-1) \approx 1.582$~\cite{BCC04,CCF06}
%
adaptive adversary
$4/3 \approx 1.333$~\cite{BCJ08}
%
%
%

%
{\em 
%
%
%
%

%
lookahead
%
%

%
\fi
\ifnum \count11 > 0
%
%
\noindent
\textbf{Related Results.}~~~
As mentioned above, 
for the $s$-bounded delay model {\em without} lookahead, 
Hajek~\cite{Haj01} showed that the competitive ratio of any deterministic algorithm is at least $(1 + \sqrt{5})/2 \approx 1.619$ in the case of $s \geq 2$. 
Independently, 
this bound was also shown in \cite{CF03,NACQ03,AZ04}. 
For $s = \infty$, 
Englert and Westermann~\cite{EW07} developed a deterministic online algorithm whose competitive ratio is at most $2 \sqrt{2} - 1 \approx 1.829$, 
which is the current best upper bound. 
For each $s = 2$~\cite{AKB01}, 3~\cite{BCC04,CCF06}, and 4~\cite{MB16}, 
an algorithm with a competitive ratio of $(1 + \sqrt{5})/2$ was designed. 
For any $s \geq 5$, 
an algorithm with a competitive ratio of larger than $(1 + \sqrt{5})/2$ but less than 2 was shown \cite{BCC04,CCF06}. 
%
%
Moreover, 
in the case where an algorithm must decide which packet to transmit on the basis of the current buffer situation, 
called the {\ memoryless} case, 
some results were shown \cite{BCC04,CCF06,EW07}. 
%
%
The {\em agreeable deadline} variant has also been studied. 
In this variant, 
the larger the release times of packets are, 
the larger their deadlines are. 
%
Specifically, 
for any packets $p$ and $p'$, 
$d(p) \leq d(p')$ 
if $r(p) < r(p')$. 
The lower bound of $(1 + \sqrt{5})/2$ by Hajek~\cite{Haj01} is applicable to this variant. 
Li et~al.~\cite{LSS05,JLS12} displayed an optimal algorithm, 
whose competitive ratio matches the lower bound.  
The case in which for any packet $p$, 
$d(p) - r(p) + 1 = s$ has also been studied, 
called the {\em $s$-uniform delay} variant,
which is a specialized variant of the agreeable deadline variant. 
The current best upper and lower bounds for this variant are 
$(1 + \sqrt{5})/2$~\cite{LSS05,JLS12} and $1.377$~\cite{CJST07}, respectively. 
%

%
The research on randomized algorithms for the bounded delay buffer management problem has also been conducted extensively \cite{CF03,BCC04,CCF06,BCJ08,Jez09a,Jez10,JLS12,Jez13}. 
In the case in which $s$ is general, 
the current best upper and lower bounds are $e/(e-1) \approx 1.582$~\cite{BCC04,CCF06,Jez13} and $5/4 = 1.25$~\cite{CF03}, respectively, against an oblivious adversary were shown. 
Upper and lower bounds of $e/(e-1)$~\cite{BCJ08,Jez13} and 
$4/3 \approx 1.333$~\cite{BCJ08}, respectively, against an adaptive adversary were shown. 
For any fixed $s$, 
lower bounds are the same with the bounds in the case in which $s$ is general while 
upper bounds are $1/(1 - (1 - \frac{1}{s})^s)$~\cite{Jez13} against the both adversaries. 
A generalization of the bounded delay buffer management problem has been studied, 
called the {\em weighted item collection} problem \cite{BCD13a,BCD13b,Jez13}. 
In this problem, 
an online algorithm does not know the deadline of each packet but knows the relative order of the deadlines of packets. 
%
%
Many other variants of the buffer management problem have been studied extensively 
(see e.g. \cite{MG10,SN15}). 
\fi
%

\section{Model Description} \label{subsec:model}
\ifnum \count10 > 0
%
%
lookahead
%
%
%
{\em 
%
%
%
%
%
%
%
%
%
2-bounded
$d(p) - r(p) \leq 1$
%
%
%
%
%
%
%
$d(p) = t$
%

%
%
%
$V_{A}(\sigma)$
$OPT$
%
$V_{OPT}(\sigma) \leq V_{ON}(\sigma) c$
%

%
$OPT$
%
%

%
\fi
\ifnum \count11 > 0
%
%
We formally give the definition of the 2-bounded delay buffer management problem with lookahead, 
which is addressed in this paper. 
An {\em input} of this problem is a sequence of phases. 
Time begins with zero and a phase occurs at an integer time. 
Each phase consists of three subphases. 
The first occurring subphase is the {\em arrival subphase}. 
%
At an arrival subphase, arbitrary many packets can arrive to a buffer. 
The buffer has no capacity limit and hence, all arriving packets can always be accepted to the buffer. 
A packet $p$ is characterized by the {\em release time}, {\em deadline} and {\em value}, denoted by $r(p)$, $d(p)$ and $v(p)$ respectively. 
Arrival times and deadlines are non-negative integers and values are positive reals. 
$d(p) - r(p) \leq 1$ holds 
because we focus on 2-bounded instances. 
The second subphase is the {\em transmission phase}. 
%
At a transmission subphase, 
an algorithm can transmit at most one packet from its buffer if any packet. 
At the transmission subphase at a time $t$, 
the algorithm can obtain the information about packets arriving at time $t+1$ using the lookahead ability. 
%
%
The third subphase is the {\em expiration subphase}. 
%
At an expiration subphase, 
a packet which has reached its deadline is discarded from its buffer. 
That is, 
at the expiration subphase at a time $t$, 
all the packets $p$ in the buffer such that $d(p) = t$ are discarded. 
The {\em profit} of an algorithm is the sum of the values of packets transmitted by the algorithm. 
The objective of this problem is to maximize the gained profit. 
Let $V_{A}(\sigma)$ denote the profit of an algorithm $A$ for an input $\sigma$. 
Let $OPT$ be an optimal offline algorithm. 
We say that the competitive ratio of an online algorithm $ON$ is at most $c$ 
if for any input $\sigma$, 
$V_{OPT}(\sigma) \leq V_{ON}(\sigma) c$. 
For ease of analysis, 
we assume that when $OPT$ does not store any packet in its buffer, 
the input is over. 
It is easy to see that 
this assumption does not affect the performance analysis of an algorithm. 
\fi
%

\section{Matching  Upper Bound} \label{UB}
\ifnum \count10 > 0
%
%

%
\fi
\ifnum \count11 > 0
%
%

%
\fi
%

%
\subsection{Notation and Definitions for Algorithm} \label{sec:alg}
\ifnum \count10 > 0
%
%
%
%
$B(t)$
$t > r(p)$
${CP}$
$OPT^{*}(t, t', t'')$
$OPT^{*}(t, t', t'')$
%
$\sigma$
$t$
%
$OPT^{*}(t, t', t'')$
%
%
$OPT^{*}(t, t', t'')$
%
%
$P(t, t', t'')$
$[t, t'']$
%
%
%
\begin{equation} \label{eq:sec.alg.1}
	P(t, t', t') \subseteq P(t, t'+1, t'+1)
\end{equation}
\begin{equation} \label{eq:sec.alg.2}
	P(t, t', t') \subseteq P(t, t', t'+1)
\end{equation}
\begin{equation} \label{eq:sec.alg.3}
	P(t+1, t', t') \subseteq P(t, t', t')
\end{equation}
%
%
\[
	m_{i}(t) = P(t, t+i, t+i) \backslash P(t, t+i-1, t+i-1)
\]
\[
	q_{i}(t) = P(t, t+i, t+i+1) \backslash P(t, t+i, t+i)
\]
%
\[
	P(t, t-1, t-1) = \varnothing
\]
%
%
%
\[
	m_{i} = m_{i}(t)
\]
\[
	m'_{i} = m_{i}(t-1)
\]
\[
	m''_{i} = m_{i}(t-2)
\]
\[
	q_{i} = q_{i}(t)
\]
\[
	q'_{i} = q_{i}(t-1)
\]
\[
	q''_{i} = q_{i}(t-2)
\]
%
%
\[
	R = \frac{1 + \sqrt{17}}{4}
\]
\[
	\alpha = \frac{ -3 + \sqrt{17} }{ 2 }
\]
%
%
%
$s_{t'} = ${\tt null}
%
%
(Cases~1.2.3.4
%
%

%
\fi
\ifnum \count11 > 0
%
%
We give definitions before defining our algorithm {\sc CompareWithPartialOPT} ($CP$). 
For any integer time $t$, 
$B(t)$ denotes the set of packets in ${CP}$'s buffer immediately before the arrival subphase at time $t$. 
That is, 
each packet $p$ in the set is not transmitted before $t$, 
$t > r(p)$ and $t \leq d(p)$. 
For integer times $t, t'(\geq t)$ and $t''(\geq t')$ and an input $\sigma$, 
let $OPT^{*}(t, t', t'')$ be an offline algorithm 
such that 
if 
the packets in $OPT^{*}(t, t', t'')$'s buffer immediately before the arrival subphase at time $t$ is equal to those in $B(t)$, 
and 
the subinput of $\sigma$ during time $[t, t']$ is given to $OPT^{*}(t, t', t'')$, 
that is, 
packets $p$ such that $r(p) \in [t, t']$ arrive to $OPT^{*}(t, t', t'')$'s buffer during time $[t, t']$, 
then 
$OPT^{*}(t, t', t'')$ is allowed to transmit packets only from time $t$ to $t''$ inclusive, and 
chooses the packets whose total profit is maximized. 
If there exist packets with the same value in $OPT^{*}(t, t', t'')$'s buffer, 
$OPT^{*}(t, t', t'')$ follows a fixed tie breaking rule. 
Also, 
$P(t, t', t'')$ denotes the set of $t'' - t + 1$ packets transmitted by $OPT^{*}(t, t', t'')$ during time $[t, t'']$. 
Note that for any $t$ and $t' (\geq t)$, 
the following relations hold 
because of the optimality of packets transmitted by $OPT^{*}(t, t', t'')$ during time $[t, t'']$: 
\begin{equation} \label{eq:sec.alg.1}
	P(t, t', t') \subseteq P(t, t'+1, t'+1)
\end{equation}
\begin{equation} \label{eq:sec.alg.2}
	P(t, t', t') \subseteq P(t, t', t'+1)
\end{equation}
and 
\begin{equation} \label{eq:sec.alg.3}
	P(t+1, t', t') \subseteq P(t, t', t'). 
\end{equation}
We define for any $t$ and $i$, 
\[
	m_{i}(t) = P(t, t+i, t+i) \backslash P(t, t+i-1, t+i-1)
\]
and 
\[
	q_{i}(t) = P(t, t+i, t+i+1) \backslash P(t, t+i, t+i). 
\]
Also, we define 
\[
	P(t, t-1, t-1) = \varnothing. 
\]
Furthermore, 
we describe each value in the algorithm definition for ease of presentation as follows: 
\[
	m_{i} = m_{i}(t)
\]
\[
	m'_{i} = m_{i}(t-1)
\]
\[
	m''_{i} = m_{i}(t-2)
\]
\[
	q_{i} = q_{i}(t)
\]
\[
	q'_{i} = q_{i}(t-1)
\]
and 
\[
	q''_{i} = q_{i}(t-2). 
\]
In addition, 
\[
	R = \frac{1 + \sqrt{17}}{4}
\]
and 
\[
	\alpha = \frac{ -3 + \sqrt{17} }{ 2 }. 
\]
${CP}$ uses the internal variable $s_{t}$ for holding the name of a packet which ${CP}$ transmits at a time $t$. 
$s_{t'} = ${\tt null} holds at first for any integer $t'$. 
${CP}$ uses the two constants {\tt tmp1} and {\tt tmp2} 
if at time $t$, 
${CP}$ cannot decide which packet to transmit at $t+1$ in Cases~1.2.3.4 and 2.2.2.3. 
On the other hand, 
once the name of a packet is set to $s_{t+1}$ at time $t$, 
${CP}$ certainly transmits the packet at $t+1$. 
\fi
%

%
%
\ifnum \count10 > 0
%
%
%
%
%
$m$
$q$
%
%
$ON$
%
%
$ON$
%
%
%
(1) 
$v(q)$
$ON$
$\beta < 1$
(2) 
$ON$
%
%
(1) 
$t+1$
%
%
%
(2) 
$t+1$
%
%
%
(1)
$P(t, t, t+1) = \{ q, m \}$
(2)
$P(t, t, t) = \{ m \}$
%
%
$ON$
$OPT$
%

%
2-bouded delay with lookahead
%
%
%
%
$ON$
(Cases~1.1, 1.2, 1.1.2)%
$t$
%
$d(q_1) = t$
%
$t+1$
$v(m_0) \geq v(m_1)$
%
%
$v(q_1)$
%
%
lookahead
$t+1$
$\alpha$
(Case~1.2.3.1) 
$ON$
$t+2$
$ON$
$OPT$
%
%
(Case~1.2.3.2) 
$ON$
$t+2$
$OPT$
%

%
$t+1$
$v(m_0) < v(m_1)$
$d(m_1) = t+2$
%
%
(Case~1.2.3.3) 
$ON$
$t+2$
$OPT$
%
%
(Case~1.2.3.4) 
$ON$
$t+1$
$t+1$
%

%
$t+1$
$t+1$
$OPT$
$t+1$
%
%
$t+2$
(Case~2.2.2.3)%
%
%

%

%
\fi
\ifnum \count11 > 0
%
%

%
\fi
%

\subsection{Algorithm} \label{sec:alg}
\ifnum \count10 > 0
%
%
\noindent\vspace{-1mm}\rule{\textwidth}{0.5mm} 
\vspace{-3mm}
{\bf {\sc CompareWithPartialOPT} ($CP$)}\\
\rule{\textwidth}{0.1mm}
%
{\bf 
%
%
$s_{t}$
%
%
%
\hspace*{0mm}
		{\bf\boldmath Case~1 ($s_t = ${\tt null}
\hspace*{2mm}
			{\bf\boldmath Case~1.1 ($d(m_{0}) = t$
				$s_{t} := m_{0}$
\hspace*{2mm}
			{\bf\boldmath Case~1.2 ($d(m_{0}) \ne t$
\hspace*{4mm}
				{\bf\boldmath Case~1.2.1 ($d(m_{1}) = t$
					$s_{t} := m_{1}$
\hspace*{4mm}
				{\bf\boldmath Case~1.2.2 ($d(m_{1}) = t + 1$
					$s_{t} := m_{0}$
\hspace*{4mm}
				{\bf\boldmath Case~1.2.3 ($d(m_{1}) \ne t + 1$
\hspace*{6mm}
				{\bf\boldmath Case~1.2.3.1 ($v(m_{0}) \geq v(m_{1})$
					$s_{t} := q_{1}$
\hspace*{6mm}
				{\bf\boldmath Case~1.2.3.2 ($v(m_{0}) \geq v(m_{1})$
					$s_{t} := m_{0}$
\hspace*{6mm}
				{\bf\boldmath Case~1.2.3.3 ($v(m_{0}) < v(m_{1})$
					$s_{t} := m_{0}$
\hspace*{6mm}
				{\bf\boldmath Case~1.2.3.4 ($v(m_{0}) < v(m_{1})$
					$s_{t} := q_{1}$
%
%
\hspace*{0mm}
		{\bf\boldmath Case~2 ($s_t = ${\tt tmp1}
\hspace*{2mm}
				{\bf\boldmath Case~2.1 ($\frac{ v(m'_{0}) + v(m'_{1}) + v(m'_{2}) }{ v(q'_{1}) + v(m'_{0}) + v(m'_{1})} \leq R$
					$s_{t} := m'_{0}$
\hspace*{2mm}
				{\bf\boldmath Case~2.2 ($\frac{ v(m'_{0}) + v(m'_{1}) + v(m'_{2})}{ v(q'_{1}) + v(m'_{0}) + v(m'_{1})} > R$
\hspace*{4mm}
				{\bf\boldmath Case~2.2.1 ($d(m'_{2}) = t + 1$
					$s_{t} := m'_{1}$
\hspace*{4mm}
				{\bf\boldmath Case~2.2.2 ($d(m'_{2}) \ne t + 1$
\hspace*{6mm}
				{\bf\boldmath Case~2.2.2.1 ($q'_{2} \ne q'_{1}$
					$s_{t} := m'_{1}$
\hspace*{6mm}
				{\bf\boldmath Case~2.2.2.2 ($q'_{2} = q'_{1}$
					$s_{t} := m'_{1}$
\hspace*{6mm}
				{\bf\boldmath Case~2.2.2.3 (
					$s_{t} := m'_{0}$
%
%
\hspace*{0mm}
		{\bf\boldmath Case~3 ($s_t = ${\tt tmp2}
\hspace*{2mm}
				{\bf\boldmath Case~3.1 ($\frac{ v(m''_{0}) + v(m''_{1}) + v(m''_{2}) + v(m''_{3}) }{ v(q''_{1}) + v(m''_{0}) + v(m''_{1}) + v(m''_{2}) } \leq R$
					$s_{t} := m''_{1}$
\hspace*{2mm}
				{\bf\boldmath Case~3.2 ($\frac{ v(m''_{0}) + v(m''_{1}) + v(m''_{2}) + v(m''_{3}) }{ v(q''_{1}) + v(m''_{0}) + v(m''_{1}) + v(m''_{2}) } > R$
\hspace*{4mm}
				{\bf\boldmath Case~3.2.1 ($d(m''_{3}) = t + 1$
					$s_{t} := m''_{2}$
\hspace*{4mm}
				{\bf\boldmath Case~3.2.2 ($d(m''_{3}) \ne t + 1$
					$s_{t} := m''_{2}$
\hspace*{4mm}
				{\bf\boldmath Case~3.2.3 ($d(m''_{3}) \ne t + 1$
					$s_{t} := m''_{2}$
%
\rule{\textwidth}{0.1mm}
\fi
\ifnum \count11 > 0
%
%
\noindent\vspace{-1mm}\rule{\textwidth}{0.5mm} 
\vspace{-3mm}
{\bf {\sc CompareWithPartialOPT} ($CP$)}\\
\rule{\textwidth}{0.1mm}
%
{\bf Initialize:} 
For any integer time $t'$, 
$s_{t'} := ${\tt null}. \\
Consider the transmission subphase at a time $t$. 
If the buffer stores no packets, 
do nothing. 
Otherwise, 
do one of the following three cases and then 
transmit the packet whose name is set to $s_{t}$. \\
%
%
\hspace*{0mm}
		{\bf\boldmath Case~1 ($s_t = ${\tt null}):}\\
\hspace*{2mm}
			{\bf\boldmath Case~1.1 ($d(m_{0}) = t$):} 
				$s_{t} := m_{0}$. \\
\hspace*{2mm}
			{\bf\boldmath Case~1.2 ($d(m_{0}) \ne t$):}\\
\hspace*{4mm}
				{\bf\boldmath Case~1.2.1 ($d(m_{1}) = t$):} 
					$s_{t} := m_{1}$ and $s_{t+1} := m_{0}$. \\
\hspace*{4mm}
				{\bf\boldmath Case~1.2.2 ($d(m_{1}) = t + 1$):}
					$s_{t} := m_{0}$ and $s_{t+1} := m_{1}$. \\
\hspace*{4mm}
				{\bf\boldmath Case~1.2.3 ($d(m_{1}) \ne t + 1$):}\\
\hspace*{6mm}
				{\bf\boldmath Case~1.2.3.1 ($v(m_{0}) \geq v(m_{1})$ and $v(q_{1}) \geq \alpha v(m_{1})$):}
					$s_{t} := q_{1}$ and $s_{t+1} := m_{0}$. \\
\hspace*{6mm}
				{\bf\boldmath Case~1.2.3.2 ($v(m_{0}) \geq v(m_{1})$ and $v(q_{1}) < \alpha v(m_{1})$):}
					$s_{t} := m_{0}$ and $s_{t+1} := m_{1}$. \\
\hspace*{6mm}
				{\bf\boldmath Case~1.2.3.3 ($v(m_{0}) < v(m_{1})$ and $\frac{ v(q_{1}) + v(m_{0}) + v(m_{1})}{v(m_{0}) + v(m_{1})} \leq R$):}
					$s_{t} := m_{0}$ and $s_{t+1} := m_{1}$. \\
\hspace*{6mm}
				{\bf\boldmath Case~1.2.3.4 ($v(m_{0}) < v(m_{1})$ and $\frac{v(q_{1}) + v(m_{0}) + v(m_{1})}{v(m_{0}) + v(m_{1})} > R$):}
					$s_{t} := q_{1}$ and $s_{t+1} :=$ {\tt tmp1}. \\
%
%
\hspace*{0mm}
		{\bf\boldmath Case~2 ($s_t = ${\tt tmp1}):}\\
\hspace*{2mm}
				{\bf\boldmath Case~2.1 ($\frac{ v(m'_{0}) + v(m'_{1}) + v(m'_{2})}{ v(q'_{1}) + v(m'_{0}) + v(m'_{1})} \leq R$):}
					$s_{t} := m'_{0}$ and $s_{t+1} := m'_{1}$. \\
\hspace*{2mm}
				{\bf\boldmath Case~2.2 ($\frac{ v(m'_{0}) + v(m'_{1}) + v(m'_{2})}{ v(q'_{1}) + v(m'_{0}) + v(m'_{1})} > R$):}\\
\hspace*{4mm}
				{\bf\boldmath Case~2.2.1 ($d(m'_{2}) = t + 1$):} 
					$s_{t} := m'_{1}$ and $s_{t+1} := m'_{2}$. \\
\hspace*{4mm}
				{\bf\boldmath Case~2.2.2 ($d(m'_{2}) \ne t + 1$):}\\
\hspace*{6mm}
				{\bf\boldmath Case~2.2.2.1 ($q'_{2} \ne q'_{1}$):}
					$s_{t} := m'_{1}$. \\
\hspace*{6mm}
				{\bf\boldmath Case~2.2.2.2 ($q'_{2} = q'_{1}$ and $\frac{ v(q'_{2}) + v(m'_{0}) + v(m'_{1}) + v(m'_{2})}{ v(q'_{1}) + v(m'_{1}) + v(m'_{2})} \leq R$):}
					$s_{t} := m'_{1}$ and $s_{t+1} := m'_{2}$. \\
\hspace*{6mm}
				{\bf\boldmath Case~2.2.2.3 (Otherwise):}
					$s_{t} := m'_{0}$ and $s_{t+1} :=$ {\tt tmp2}. \\
%
%
\hspace*{0mm}
		{\bf\boldmath Case~3 ($s_t = ${\tt tmp2}):}\\
\hspace*{2mm}
				{\bf\boldmath Case~3.1 ($\frac{ v(m''_{0}) + v(m''_{1}) + v(m''_{2}) + v(m''_{3}) }{ v(q''_{1}) + v(m''_{0}) + v(m''_{1}) + v(m''_{2}) } \leq R$):}
					$s_{t} := m''_{1}$ and $s_{t+1} := m''_{2}$. \\
\hspace*{2mm}
				{\bf\boldmath Case~3.2 ($\frac{ v(m''_{0}) + v(m''_{1}) + v(m''_{2}) + v(m''_{3}) }{ v(q''_{1}) + v(m''_{0}) + v(m''_{1}) + v(m''_{2}) } > R$):}\\
\hspace*{4mm}
				{\bf\boldmath Case~3.2.1 ($d(m''_{3}) = t + 1$):}
					$s_{t} := m''_{2}$ and $s_{t+1} := m''_{3}$. \\
\hspace*{4mm}
				{\bf\boldmath Case~3.2.2 ($d(m''_{3}) \ne t + 1$ and $q''_{3} \ne q''_{1}$):}
					$s_{t} := m''_{2}$. \\
\hspace*{4mm}
				{\bf\boldmath Case~3.2.3 ($d(m''_{3}) \ne t + 1$ and $q''_{3} = q''_{1}$):}
					$s_{t} := m''_{2}$ and $s_{t+1} := m''_{3}$. \\
%
\rule{\textwidth}{0.1mm}
\fi
%

\subsection{Overview of the Analysis} \label{sec:overview}
\ifnum \count10 > 0
%
%
%
%
${CP}$
${CP}$
$k$
$T_{i} = (t_{i}, t'_{i})$
$t_{i} \leq t'_{i}$
$t_{1} = 0$
$t'_{k} = \tau$
$t_{j} = t'_{j-1} + 1$
%
%
$T_{i}$
${CP}$
%
$T_{i} = (t, t')$
%
\begin{itemize}
	\itemsep=-2.0pt
	\setlength{\leftskip}{10pt}
	\item
		$t$
		$t' = t$
	\item
		$t$
		$t' = t+1$
	\item
		$t$
		$t+1$
		$t' = t+2$
	\item
		$t$
		$t+1$
		$t' = t+1$
	\item
		$t$
		$t+1$
		$t+2$
		$t' = t+3$
	\item
		$t$
		$t+1$
		$t+2$
		$t' = t+2$
\end{itemize}
%

%
{\em $2_{t}$-packet}
%
%
$OPT$
$\tau'$
${CP}$
$T_{i}$
$T'_{i}$
$T'_{i}$
$T_{i} = (t, t')$
%
$T'_{i} = (\hat{t}, \hat{t}')$
%
${CP}$
$OPT$
$\hat{t} = t+1$
$\hat{t} = t$
%
%
${CP}$
$OPT$
$\hat{t}' = t'+1$
$\hat{t}' = t'$
%

%
%

%
\fi
\ifnum \count11 > 0
%
%
Consider a given input $\sigma$. 
Let $\tau$ be the time at which ${CP}$ transmits the last packet. 
We partition the time sequence $[0, \tau]$ into $k$ sequences $T_{i} (i = 1, \ldots, k)$ to evaluate the competitive ratio of ${CP}$, 
in which 
$k$ depends on $\sigma$ and 
if $T_{i} = (t_{i}, t'_{i})$, 
then $t_{i} \leq t'_{i}$, 
$t_{1} = 0$, 
$t'_{k} = \tau$ and 
for any $j = 2, \ldots, k$, 
$t_{j} = t'_{j-1} + 1$. 
The size of each $T_{i}$ depends on which case ${CP}$ executes at each time. 
Specifically, 
it is defined as follows: 
Suppose that $T_{i} = (t, t')$ and then 
%
\begin{itemize}
	\itemsep=-2.0pt
	\setlength{\leftskip}{10pt}
	\item
		If Case~1.1 is executed at $t$, 
		then $t' = t$. 
	\item
		If Case~1.2.1, 1.2.2, 1.2.3.1, 1.2.3.2 or 1.2.3.3 is executed at $t$, 
		then $t' = t+1$. 
	\item
		If Case~1.2.3.4 is executed at $t$ 
		and Case~2.1, 2.2.1 or 2.2.2.2 at $t+1$, 
		then $t' = t+2$. 
	\item
		If Cases~1.2.3.4 and 2.2.2.1 are executed at $t$ and $t+1$, respectively, 
		then $t' = t+1$. 
	\item
		If Cases~1.2.3.4 and 2.2.2.3 are executed at $t$ and  $t+1$, respectively, and 
		Case~3.1, 3.2.1 or  3.2.3 is executed at $t+2$, 
		then $t' = t+3$. 
	\item
		If Cases~1.2.3.4 and 2.2.2.3 are executed at $t$ and  $t+1$, respectively, and 
		Case~3.2.2 is executed at $t+2$, 
		then $t' = t+2$. 
\end{itemize}
%
%
For a time $t$, 
a packet whose release time is $t$ and deadline is $t+1$ is called a {\em $2_{t}$-packet}. 
On the other hand, 
we also partition the time sequence $[0, \tau']$ into $k$ time sequences $T'_{i} \hspace{1mm} (i = 1, \ldots, k)$, 
in which 
$\tau'$ is the time at which $OPT$ transmits the last packet. 
We will compare the total value of packets transmitted by ${CP}$ during the time $T_{i}$ with that by $OPT$ during the time $T'_{i}$. 
$T'_{i}$ is defined as follows: 
Suppose that $T_{i} = (t, t')$. 
Then, 
we define $T'_{i} = (\hat{t}, \hat{t}')$, 
in which 
if ${CP}$ transmits a $2_{t-1}$-packet $p$ at time $t-1$ and 
$OPT$ transmits $p$ at time $t$, 
then $\hat{t} = t+1$ and otherwise, 
$\hat{t} = t$. 
Moreover, 
if ${CP}$ transmits a $2_{t'}$-packet $p'$ and 
$OPT$ transmits $p'$ at time $t'+1$, 
then $\hat{t}' = t'+1$. 
Otherwise, 
$\hat{t}' = t'$. 
We give the lemma about $T'_{i}$. 
\fi
%
\begin{LMA}\label{LMA:L0}
	\ifnum \count10 > 0
	%
	%
	%
	
	%
	\fi
	\ifnum \count11 > 0
	%
	%
	A time $t \in [0, \tau']$ is contained in some $T'_{j}$. 
	\fi
\end{LMA}
\begin{proof}
	\ifnum \count10 > 0
	%
	%
	%
	%
	%
	$t$
	%
	%
	$t$
	$t$
	$T'$
	$t$
	%
	%
	$t$
	%
	$t$
	%
	%
	$OPT$
	$2_{t-1}$-packet
	$p$
	%
	%
	${CP}$
	${CP}$
	${CP}$
	$T_{i}$
	%
	$T'$
	$t$
	%
	%
	%
	%
	
	%
	\fi
	\ifnum \count11 > 0
	%
	%
	We prove this lemma by contradiction and assume that 
	a time $t \in [0, \tau']$ is not contained in any $T'_{j}$. 
	By the assumption of a given input, 
	$OPT$ transmits a packet at $t$. 
	If ${CP}$ transmits a packet at $t$, 
	then $t$ is contained in some $T_{i}$ and 
	$t$ is contained in either $T'_{i}$ or $T'_{i-1}$ by the definition of $T'$. 
	Thus, 
	${CP}$ does not store any packet at $t$ and 
	does not transmit a packet. 
	That is, 
	no packets arrive at $t$. 
	A packet $p$ transmitted by $OPT$ at $t$ is a $2_{t-1}$-packet. 
	If not, 
	${CP}$ can transmit $p$ at $t$. 
	Hence, 
	${CP}$ transmits $p$ at time $t-1$. 
	By the definition of ${CP}$, 
	${CP}$ executes Case~1.2.3.2, 1.2.3.3, 2.2.2.2 or 3.2.3 and 
	transmits $p$ at $t-1$, 
	which is the last time of $T_{i}$. 
	Therefore, 
	$t$ is contained in $T'_{i}$ by the definition of $T'$, 
	which contradicts the above assumption. 
	\fi
\end{proof}
\ifnum \count10 > 0
%
%
$V_{i}$
$T_{i}$
%
\[
	V_{{CP}}(\sigma) = \sum_{i = 1}^{k} V_{i}
\]
\[
	V_{OPT}(\sigma) \leq \sum_{i = 1}^{k} V'_{i}
\]
%
\[
	\frac{ V_{OPT}(\sigma) }{ V_{{CP}}(\sigma) } 
		 \leq \frac{ \sum_{i = 1}^{k} V'_{i} }{ \sum_{i = 1}^{k} V_{i} }
		 \leq \max_{ i \in [1, k] } \left \{ \frac{ V'_{i} }{ V_{i} } \right \} 
\]
\[
	\frac{ V'_{i} }{ V_{i} }
		 \leq R
\]
%
%
%

%
\fi
\ifnum \count11 > 0
%
%
For any $i (\in [1, k])$, 
let $V_{i}$ ($V'_{i}$) denote the total value of packets transmitted by ${CP}$ ($OPT$) 
during $T_{i}$ ($T'_{i}$). 
By definition, 
\[
	V_{{CP}}(\sigma) = \sum_{i = 1}^{k} V_{i}. 
\]
By Lemma~\ref{LMA:L0}, 
\[
	V_{OPT}(\sigma) \leq \sum_{i = 1}^{k} V'_{i}. 
\]
Since 
\[
	\frac{ V_{OPT}(\sigma) }{ V_{{CP}}(\sigma) } 
		 \leq \frac{ \sum_{i = 1}^{k} V'_{i} }{ \sum_{i = 1}^{k} V_{i} }
		 \leq \max_{ i \in [1, k] } \left \{ \frac{ V'_{i} }{ V_{i} } \right \}, 
\]
we have using Lemma~\ref{LMA:L3}, 
\[
	\frac{ V'_{i} }{ V_{i} }
		 \leq R. 
\]
Therefore, 
we have the following theorem: 
\fi
\begin{THM}\label{thm:upcr}
	\ifnum \count10 > 0
	${CP}$
	\fi
	\ifnum \count11 > 0
	The competitive ratio of ${CP}$ is at most $(1 + \sqrt{17})/4$. 
	\fi
\end{THM}
%

\subsection{Analysis} \label{sec:analysis}


%
\ifnum \count10 > 0
%
%
%
%
$V(t, t', t'')$
$P(t, t', t'')$
$V(t, t', t'') = \sum_{ p \in P(t, t', t'')} v(p)$
%

%
\fi
\ifnum \count11 > 0
%
%
To show Lemma~\ref{LMA:L3}, 
we first prove the following lemmas. 
Let $V(t, t', t'')$ denote the total value of packets in $P(t, t', t'')$. 
That is, 
$V(t, t', t'') = \sum_{ p \in P(t, t', t'')} v(p)$. 
\fi
%

%
\begin{LMA}\label{LMA:L1}
	\ifnum \count10 > 0
	%
	%
	$T_{i} = (t, t')$
	%
	$V'_{i} \leq V(t, t', t')$
	%
	
	%
	\fi
	\ifnum \count11 > 0
	%
	%
	For integer times $t$ and $t' (\geq t)$, 
	suppose that $T_{i} = (t, t')$ and 
	a packet that ${CP}$ transmits at $t'$ is not a $2_{t'}$-packet. 
	Then, 
	$V'_{i} \leq V(t, t', t')$. 
	\fi
\end{LMA}
\begin{proof}
	\ifnum \count10 > 0
	%
	%
	$T_{i} = (t, t')$
	%
	%
	%
	\begin{equation} \label{eq:LMA.L1.1}
		T'_{i} = (t, t')
	\end{equation}
	%
	%
	$t-1$
	%
	%
	2-bounded
	%
	%
	$OPT^{*}(t, t', t')$
	\[
		V'_{i} \leq V(t, t', t')
	\]
	%
	
	%
	${CP}$
	$2_{t-1}$-packet $p$
	%
	%
	$OPT$
	%
	%
	\begin{equation} \label{eq:LMA.L1.2}
		T'_{i} = (t, t')
	\end{equation}
	%
	${CP}$
	$OPT$
	%
	$OPT$
	$p$
	%
	%
	$OPT^{*}(t, t', t')$
	\[
		V'_{i} \leq V(t, t', t')
	\]
	%
	
	%
	$OPT$
	%
	\[
		T'_{i} = (t+1, t')
	\]
	%
	\[
		V'_{i} \leq V(t+1, t', t') \leq V(t, t', t')
	\]
	%
	%
	
	%
	\fi
	\ifnum \count11 > 0
	%
	%
	Suppose that $T_{i} = (t, t')$. 
	First, 
	we consider ${CP}$ does not transmit a $2_{t-1}$-packet at time $t-1$. 
	Then, by definition, 
	\begin{equation} \label{eq:LMA.L1.1}
		T'_{i} = (t, t'). 
	\end{equation}
	Furthermore, 
	all the $2_{t-1}$-packet, which arrive at $t-1$, are stored in ${CP}$'s buffer at time $t$. 
	Hence, 
	since we discussed the 2-bounded instance in this paper, 
	all the packets in $OPT$'s buffer at $t$ are stored in ${CP}$'s buffer. 
	Thus, 
	by the optimality of $OPT^{*}(t, t', t')$ from Eq.~(\ref{eq:LMA.L1.1}), 
	\[
		V'_{i} \leq V(t, t', t'). 
	\]
	Next, 
	we consider the case in which 
	${CP}$ transmits a $2_{t-1}$-packet $p$ at $t-1$. 
	First, 
	let us consider the case in which 
	$OPT$ does not transmit $p$ at $t$. 
	By definition, 
	\begin{equation} \label{eq:LMA.L1.2}
		T'_{i} = (t, t'). 
	\end{equation}
	At time $t$, 
	$OPT$'s buffer may store $p$ but 
	${CP}$'s buffer does not. 
	However, 
	it does not affect a packet which $OPT$ transmits at $t$ whether $p$ is stored in $OPT$'s buffer 
	because $OPT$ does not transmit $p$ at $t$. 
	Hence, 
	by the optimality of $OPT^{*}(t, t', t')$ from Eq.~(\ref{eq:LMA.L1.2}), 
	\[
		V'_{i} \leq V(t, t', t'). 
	\]
	Second, 
	we discuss the case in which $OPT$ transmits $p$ at $t$. 
	By definition, 
	\[
		T'_{i} = (t+1, t'). 
	\]
	Therefore, 
	\[
		V'_{i} \leq V(t+1, t', t') \leq V(t, t', t'), 
	\]
	in which the second inequality follows from Eq.~(\ref{eq:sec.alg.3}). 
	\fi
\end{proof}
\ifnum \count10 > 0
%
%

%

%
\fi
\ifnum \count11 > 0
%
%

%
\fi
\begin{LMA}\label{LMA:L2}
	\ifnum \count10 > 0
	%
	%
	$T_{i} = (t, t')$
	$2_{t'}$-packet
	%
	$V'_{i} \leq V(t, t', t'+1)$
	%
	
	%
	\fi
	\ifnum \count11 > 0
	%
	%
	For integer times $t$ and $t' (\geq t)$, 
	suppose that $T_{i} = (t, t')$ and 
	a packet that ${CP}$ transmits at $t'$ is a $2_{t'}$-packet. 
	Then, 
	$V'_{i} \leq V(t, t', t'+1)$. 
	\fi
\end{LMA}
\begin{proof}
	\ifnum \count10 > 0
	%
	%
	$T_{i} = (t, t')$
	$OPT$
	%
	$T'_{i} = (t, t')$%
	$T'_{i} = (t+1, t')$
	\[
		V'_{i} \leq V(t+1, t', t') \leq V(t, t', t')
	\]
	%
	\[
		V'_{i} \leq V(t, t', t'+1)
	\]
	%
	
	%
	$OPT$
	$2_{t'}$-packet $p'$
	%
	%
	${CP}$
	$OPT$
	$T'_{i} = (t+1, t'+1)$
	%
	$T'_{i} = (t, t'+1)$
	%
	%
	\[
		V'_{i} \leq V(t, t'+1, t'+1)
	\]
	%
	$OPT$
	$t'+1$
	%
	\[
			P(t, t'+1, t'+1) = P(t, t', t'+1)
	\]
	\[
		V'_{i} \leq V(t, t', t'+1)
	\]
	%
	
	%
	\fi
	\ifnum \count11 > 0
	%
	%
	Suppose that $T_{i} = (t, t')$. 
	We have in a similar way to the proof of Lemma~\ref{LMA:L1}, either $T'_{i} = (t, t')$ or $T'_{i} = (t+1, t')$. 
	Thus, 
	\[
		V'_{i} \leq V(t+1, t', t') \leq V(t, t', t'). 
	\]
	Hence, 
	by Eq.~(\ref{eq:sec.alg.2}), 
	\[
		V'_{i} \leq V(t, t', t'+1). 
	\]
	Second, 
	we consider the case in which 
	$OPT$ transmits a $2_{t'}$-packet $p'$ at $t'+1$ which 
	${CP}$ transmits at $t'$. 
	If ${CP}$ transmits a $2_{t-1}$-packet $p$ at $t-1$
	and $OPT$ transmits $p$ at $t$, 
	then $T'_{i} = (t+1, t'+1)$. 
	Otherwise, 
	$T'_{i} = (t, t'+1)$. 
	In the either case, 
	we have in a similar way to the proof of Lemma~\ref{LMA:L1}, 
	\[
		V'_{i} \leq V(t, t'+1, t'+1). 
	\]
	Since $OPT$ transmits a packet at time $t'+1$ which arrives at time $t'$, 
	packets arriving at $t'+1$ does not matter to $OPT$. 
	Thus, 
	\[
			P(t, t'+1, t'+1) = P(t, t', t'+1), 
	\]
	which leads to 
	\[
		V'_{i} \leq V(t, t', t'+1). 
	\]
	\fi
\end{proof}
\ifnum \count10 > 0
%
%

%

%
\fi
\ifnum \count11 > 0
%
%

%
\fi
\begin{LMA}\label{LMA:L5}
	\ifnum \count10 > 0
	%
	%
	$OPT$
	\fi
	\ifnum \count11 > 0
	%
	%
	Let $t$ be an integer time. 
	If ${CP}$ executes Case~2.2.2.1 at time $t+1$, 
	then $OPT$ transmits $m_{0}(t)$ and $m_{1}(t)$ at times $t$ and $t+1$, respectively. 
	\fi
\end{LMA}
\begin{proof}
	\ifnum \count10 > 0
	%
	%
	%
	Case
	$q_{1}(t) \ne q_{2}(t)$
	%
	%
	$r(q_{2}(t)) \leq t + 1$
	%
	%
	$P$
	$P(t, t+1, t+2) = \{ q_{1}(t), m_{0}(t), m_{1}(t) \}$
	%
	$q_{2}(t)$
	$q_{2}(t)$
	%
	%
	$OPT^{*}(t, t+2, t+3)$
	$P(t, t+1, t+2)$
	%
	%
	$v(q_{1}(t)) < v(q_{2}(t))$
	%
	%
	$V(t, t+1, t+2) < v(q_{2}(t)) + v(m_{0}(t)) + v(m_{1}(t))$
	%
	%
	%
	
	%
	$r(q_{2}(t)) \geq t + 2$
	%
	%
	$P(t, t+2, t+2)$
	$OPT^{*}(t, t+2, t+2)$
	%
	$P(t, t+2, t+3)$
	$P(t, t+2, t+2)$
	$OPT^{*}(t, t+2, t+3)$
	%
	%
	Case~2.2.2
	$d(m_{2}(t)) = t+3$
	$r(m_{2}(t)) \geq t+2$
	%
	%
	$t$
	$t$
	%
	$OPT$
	$m_{0}(t)$
	\fi
	\ifnum \count11 > 0
	%
	%
	Suppose that Case~2.2.2.1 is executed at a time $t+1$. 
	Then, 
	note that $q_{1}(t) \ne q_{2}(t)$ by the condition of Case~2.2.2.1. 
	First, 
	let us consider the case in which $r(q_{2}(t)) \leq t + 1$. 
	By the definition of $P$, 
	$P(t, t+1, t+2) = \{ q_{1}(t), m_{0}(t), m_{1}(t) \}$. 
	However, 
	$q_{2}(t)$ is neither $m_{0}(t)$ nor $m_{1}(t)$ 
	by the definition of $q_{2}(t)$. 
	Also, 
	$OPT^{*}(t, t+2, t+3)$ obtains a higher profit by $q_{2}(t)$ in $P(t, t+2, t+3)$ instead of $q_{1}(t)$ in $P(t, t+1, t+2)$. 
	That is, 
	$v(q_{1}(t)) < v(q_{2}(t))$. 
	Thus, 
	$V(t, t+1, t+2) < v(q_{2}(t)) + v(m_{0}(t)) + v(m_{1}(t))$, 
	which contradicts the optimality of $OPT^{*}(t, t+1, t+2)$. 
	In the following, 
	we consider the case in which 
	$r(q_{2}(t)) \geq t + 2$. 
	$P(t, t+2, t+2)$ contains one packet which $OPT^{*}(t, t+2, t+2)$ transmits at time $t+2$. 
	Also, 
	$P(t, t+2, t+3)$ contains $q_{2}(t)$, 
	which is not in $P(t, t+2, t+2)$.	
	$OPT^{*}(t, t+2, t+3)$ transmits $q_{2}(t)$ at or after $t+2$. 
	Furthermore, 
	$r(m_{2}(t)) \geq t+2$ 
	because $d(m_{2}(t)) = t+3$ by the condition of Case~2.2.2. 
	In a similar way to the proofs of Lemmas~\ref{LMA:L1} and \ref{LMA:L2}, 
	we can show that packets in $OPT$'s buffer at $t$ is included in ones in ${CP}$'s buffer at $t$. 
	Therefore, 
	$OPT$ transmits $q_{2}(t)$ and $m_{2}(t)$ at or after $t+2$ and 
	transmits $m_{0}(t)$ and $m_{1}(t)$ at $t$ and $t+1$. 
	\fi
\end{proof}
\ifnum \count10 > 0
%
%

%

%
\fi
\ifnum \count11 > 0
%
%

%
\fi
\begin{LMA}\label{LMA:L6}
	\ifnum \count10 > 0
	%
	%
	$OPT$
	\fi
	\ifnum \count11 > 0
	%
	%
	Let $t$ be an integer time. 
	If ${CP}$ executes Case~3.2.2 at time $t+2$, 
	then $OPT$ transmits $m_{0}(t)$, $m_{1}(t)$ and $m_{2}(t)$ at times $t$, $t+1$ and $t+2$, respectively. 
	\fi
\end{LMA}
\begin{proof}
	\ifnum \count10 > 0
	%
	%
	%
	%
	%
	Case~2.2.2.3
	$q_{1}(t) = q_{2}(t)$
	Case~3.2.2
	$q_{1}(t) \ne q_{3}(t)$
	%
	%
	$r(q_{2}(t)) \leq t + 2$
	%
	%
	$P(t, t+2, t+3)$
	$OPT^{*}(t, t+2, t+3)$
	%
	%
	
	%
	$r(q_{2}(t)) \geq t + 3$
	%
	%
	$P(t, t+3, t+3)$
	$OPT^{*}(t, t+3, t+3)$
	%
	$P(t, t+3, t+4)$
	$P(t, t+3, t+3)$
	$OPT^{*}(t, t+3, t+4)$
	%
	%
	Case~2.2.2
	$d(m_{3}(t)) = t+4$
	$r(m_{3}(t)) \geq t+3$
	%
	$OPT$
	$m_{0}(t)$
	\fi
	\ifnum \count11 > 0
	%
	%
	We can show the proof of this lemma in a similar way to the proof of Lemma~\ref{LMA:L5}.
	Suppose that Case~2.2.2.3 is executed at a time $t+1$ and 
	Case~3.2.2 is executed at $t+2$. 
	Note that $q_{1}(t) = q_{2}(t)$
	by the condition of Case~2.2.2.3 and 
	$q_{1}(t) \ne q_{3}(t)$ by the condition of Case~3.2.2. 
	If $r(q_{2}(t)) \leq t + 2$, 
	then $OPT^{*}(t, t+2, t+3)$ can transmit $q_{2}(t)$ during time $[t, t+2]$ instead of $q_{1}(t)$ concerning $P(t, t+2, t+3)$, 
	which contradicts the optimality of $OPT^{*}(t, t+2, t+3)$. 
	Thus, 
	$r(q_{2}(t)) \geq t + 3$. 
	$P(t, t+3, t+3)$ contains one packet transmitted by $OPT^{*}(t, t+3, t+3)$ at time $t+3$. 
	Also, 
	$P(t, t+3, t+4)$ contains $q_{3}(t)$, 
	which is not in $P(t, t+3, t+3)$, 
	and 
	$OPT^{*}(t, t+3, t+4)$ transmits $q_{3}(t)$ at or after $t+3$. 
	Moreover, 
	since $d(m_{3}(t)) = t+4$ by the condition of Case~2.2.2, 
	$r(m_{3}(t)) \geq t+3$. 
	Therefore, 
	$OPT$ transmits $q_{3}(t)$ and $m_{3}(t)$ at or after $t+3$ and 
	transmits $m_{0}(t)$, $m_{1}(t)$ and $m_{2}$ during time $[t, t+2]$. 
	\fi
\end{proof}
\ifnum \count10 > 0
%
%
%

%
\fi
\ifnum \count11 > 0
%
%
We are ready to prove Lemma~\ref{LMA:L3}. 
\fi
\begin{LMA}\label{LMA:L3}
	\ifnum \count10 > 0
	%
	%
	$V'_{i} / V_{i} \leq R$
	%
	
	%
	\fi
	\ifnum \count11 > 0
	%
	%
	$V'_{i} / V_{i} \leq R$. 
	\fi
\end{LMA}
\begin{proof}
	\ifnum \count10 > 0
	%
	%
	%
	%
	%
	$t$
	$t+1$
	$t+1$
	$t+2$
	%
	%
	$m_{i}(t)$
	$m_{i}$
	%
	
	%
	%
	Case~1.2.3.4
	\begin{equation} \label{eq:L3.1}
		v(m_{0}) < v(m_{1})
	\end{equation}
	\begin{equation} \label{eq:L3.2}
		v(q_{1}) > (R-1)(v(m_{0}) + v(m_{1}))
	\end{equation}
	%
	%
	\begin{equation} \label{eq:L3.3}
		v(m_{2}) > (R-1)(v(m_{0}) + v(m_{1})) + R v(q_{1})
	\end{equation}
	%
	%
	\begin{equation} \label{eq:L3.5}
		v(m_{3}) > (R-1)(v(m_{0}) + v(m_{1}) + v(m_{2})) + R v(q_{1})
	\end{equation}
	%
	%
	\begin{eqnarray*} 
		v(m_{0}) 
			&>& (R - 1)(v(q_{2}) + v(m_{1}) + v(m_{2})) \\
			&=& (R - 1)(v(q_{1}) + v(m_{1}) + v(m_{2})) 
				\mbox{\hspace*{10mm} (Case~2.2.2.3
			&>& (R - 1)(v(q_{1}) + v(m_{1})) + (R-1)( (R-1)(v(m_{0}) + v(m_{1})) + R v(q_{1}) ) 
				\mbox{\hspace*{10mm} (
			&=& (R^2 - 1) v(q_{1}) + (R^2 - R) v(m_{1}) + (R^2 - 2R + 1) v(m_{0}) \\
			&>& (R^2 - 1)(R - 1)(v(m_{0} + v(m_{1})) + (R^2 - R) v(m_{1}) + (R^2 - 2R + 1) v(m_{0}) 
				\mbox{\hspace*{5mm} (
			&=& (R^3 - 3R + 2) v(m_{0}) + (R^3 - 2R + 1) v(m_{1})
	\end{eqnarray*}
	\begin{equation} \label{eq:L3.6}
		v(m_{0}) > \frac{ R^3 - 2R + 1 }{ -R^3 + 3 R - 1 } v(m_{1})
	\end{equation}
	%
	
	%
	%
	Case~1.1
	${CP}$
	$V_{j} = v(m_{0}) = V(t, t, t)$
	%
	%
	$V'_{j} \leq V(t, t, t)$
	%
	%
	$V'_{j} / V_{j} \leq 1$
	%
	
	%
	Case~1.2.1
	${CP}$
	$V_{j} = v(m_{0}) + v(m_{1}) = V(t, t+1, t+1)$
	%
	%
	$V'_{j} \leq V(t, t+1, t+1)$
	%
	%
	$V'_{j} / V_{j} \leq 1$
	%
	
	%
	Case~1.2.3.1
	Case~1.2.3.1
	${CP}$
	\[	
		V_{j} = v(q_{1}) + v(m_{0}) 
	\]
	%
	$P$
	$P(t, t+1, t+1) = \{ m_{0}, m_{1} \}$
	\[
		V'_{j} \leq V(t, t+1, t+1) = v(m_{0}) + v(m_{1})
	\]
	%
	%
	\[
		\frac{ V'_{j} }{ V_{j} } \leq \frac{ v(m_{0}) + v(m_{1}) }{ v(q_{1}) + v(m_{0}) }
			\leq \frac{ v(m_{1}) + v(m_{1}) }{ \alpha v(m_{1}) + v(m_{1}) }
			= \frac{ 2 }{ \alpha + 1 } 
			= R
	\]
	2
	Case~1.2.3.1
	$v(m_{0}) \geq v(m_{1})$
	$\alpha$
	%
	
	%
	Case~1.2.3.2
	Case~1.2.3.2
	${CP}$
	\[	
		V_{j} = v(m_{0}) + v(m_{1})
	\]
	%
	%
	$P$
	$P(t, t+1, t+1) = \{ m_{0}, m_{1} \}$
	$P(t, t+1, t+2) = \{ q_{1}, m_{0}, m_{1} \}$
	%
	%
	$OPT$
	\[
		V'_{j} \leq V(t, t+1, t+1) = v(m_{0}) + v(m_{1})
	\]
	%
	$V'_{j} / V_{j} \leq 1$
	%
	%
	$OPT$
	\[
		V'_{j} \leq V(t, t+1, t+2) = v(q_{1}) + v(m_{0}) + v(m_{1})
	\]
	%
	%
	\[
		\frac{ V'_{j} }{ V_{j} } \leq \frac{ v(q_{1}) + v(m_{0}) + v(m_{1}) }{ v(m_{0}) + v(m_{1}) }
			< \frac{ \alpha v(m_{0}) + v(m_{0}) + v(m_{0}) }{ v(m_{0}) + v(m_{0}) }
			= \frac{ \alpha + 2 }{ 2 } 
			= R
	\]
	2
	Case~1.2.3.2
	$v(m_{0}) \geq v(m_{1})$
	$\alpha$
	%
	
	%
	Case~1.2.3.3
	Case~1.2.3.3
	${CP}$
	\[	
		V_{j} = v(m_{0}) + v(m_{1})
	\]
	%
	%
	Case~1.2.3.2
	$OPT$
	\[
		V'_{j} \leq V(t, t+1, t+1) = v(m_{0}) + v(m_{1})
	\]
	%
	$V'_{j} / V_{j} \leq 1$
	%
	%
	$OPT$
	\[
		V'_{j} \leq V(t, t+1, t+2) = v(q_{1}) + v(m_{0}) + v(m_{1})
	\]
	%
	%
	\[
		\frac{ V'_{j} }{ V_{j} } \leq \frac{ v(q_{1}) + v(m_{0}) + v(m_{1}) }{ v(m_{0}) + v(m_{1}) }
			\leq R
	\]
	%
	
	%
	%
	%
	%
	Case~2.1
	${CP}$
	$t+1$
	\[
		V_{j} = v(q_{1}) + v(m_{0}) + v(m_{1})
	\]
	%
	$P(t, t+2, t+2) =  \{ m_{0}, m_{1}, m_{2} \}$
	\[
		V'_{j} \leq V(t, t+2, t+2) = v(m_{0}) + v(m_{1}) + v(m_{2})
	\]
	%
	%
	Case~2.1
	\[
		\frac{ V'_{j} }{ V_{j} } 
			\leq \frac{ v(m_{0}) + v(m_{1}) + v(m_{2}) }{ v(q_{1}) + v(m_{0}) + v(m_{1}) }
			\leq R
	\]
	%
	
	%
	Case~2.2.1
	Case~2.2.1
	${CP}$
	\[
		V_{j} = v(q_{1}) + v(m_{1}) + v(m_{2})
	\]
	%
	%
	$P(t, t+2, t+2) =  \{ m_{0}, m_{1}, m_{2} \}$
	\[
		V'_{j} \leq V(t, t+1, t+1) = v(m_{0}) + v(m_{1}) + v(m_{2})
	\]
	%
	%
	\begin{eqnarray*}
		\frac{ V'_{j} }{ V_{j} } 
			&\leq& \frac{ v(m_{0}) + v(m_{1}) + v(m_{2}) }{ v(q_{1}) + v(m_{1}) + v(m_{2}) }\\
			&<&  \frac{ v(m_{0}) + v(m_{1}) + (R-1)(v(m_{0}) + v(m_{1})) + R v(q_{1}) }{ v(q_{1}) + v(m_{1}) + (R-1)(v(m_{0}) + v(m_{1})) + R v(q_{1}) } 
				\mbox{\hspace*{10mm} (
			&=&  \frac{ R (v(m_{0}) + v(m_{1})) + R v(q_{1}) }{ R(v(q_{1}) + v(m_{1})) - v(m_{0}) + (R+1) v(q_{1}) }\\
			&<& \frac{ R (v(m_{0}) + v(m_{1})) + R(R-1)(v(m_{0}) + v(m_{1})) }{ R(v(q_{1}) + v(m_{1})) - v(m_{0}) + (R+1)(R-1)(v(m_{0}) + v(m_{1})) } \mbox{\hspace*{10mm} (
			&=& \frac{ R^2 (v(m_{0}) + v(m_{1})) }{ (R^2 + R - 1)(v(q_{1}) + v(m_{1})) - v(m_{0}) } 
			 <  \frac{ 2 R^2 v(m_{0}) }{ 2(R^2 + R - 1) v(m_{0}) - v(m_{0}) } \mbox{\hspace*{10mm} (
			&=& \frac{ 2 R^2 }{ 2 R^2 + 2 R - 3 } < \frac{ R^2 }{ R } = R
	\end{eqnarray*}
	%
	
	%
	Case~2.2.2.1
	Case~2.2.2.1
	${CP}$
	$m_{1}$
	\[
		V_{j} = v(q_{1}) + v(m_{1})
	\]
	%
	%
	\[
		V'_{j} = v(m_{0}) + v(m_{1})
	\]
	%
	%
	\begin{eqnarray*}
		\frac{ V'_{j} }{ V_{j} } 
			&\leq& \frac{ v(m_{0}) + v(m_{1}) }{ v(q_{1}) + v(m_{1}) } \\
			&<& \frac{ v(m_{0}) + v(m_{1}) }{ (R-1)(v(m_{0}) + v(m_{1})) + v(m_{1}) } \mbox{\hspace*{10mm} (
			&<& \frac{ v(m_{0}) + v(m_{0}) }{ (R-1)(v(m_{0}) + v(m_{0})) + v(m_{0}) } \mbox{\hspace*{10mm} (
			&=& \frac{ 2 }{ R } = \sqrt{5} - 1 < R
	\end{eqnarray*}
	%
	
	%
	Case~2.2.2.2
	Case~2.2.2.2
	${CP}$
	\[	
		V_{j} = v(q_{1}) + v(m_{1}) + v(m_{2})
	\]
	%
	%
	$P$
	$P(t, t+2, t+2) = \{ m_{0}, m_{1}, m_{2} \}$
	Case~2.2.2.2
	$q_{2} = q_{1}$
	$P(t, t+2, t+3) = \{ q_{2}, m_{0}, m_{1}, m_{2} \} = \{ q_{1}, m_{0}, m_{1}, m_{2} \}$
	%
	%
	$OPT$
	\[
		V'_{j} \leq V(t, t+2, t+2) = v(m_{0}) + v(m_{1}) + v(m_{2})
	\]
	%
	\[
		\frac{ V'_{j} }{ V_{j} } 
			\leq \frac{ v(m_{0}) + v(m_{1}) + v(m_{2})}{ v(q_{1}) + v(m_{1}) + v(m_{2}) }
			 <   \frac{ v(q_{1}) + v(m_{0}) + v(m_{1}) + v(m_{2})}{ v(q_{1}) + v(m_{1}) + v(m_{2}) }
	\]
	%
	%
	$OPT$
	\[
		V'_{j} \leq V(t, t+2, t+3) = v(q_{1}) + v(m_{0}) + v(m_{1}) + v(m_{2})
	\]
	%
	%
	\[
		\frac{ V'_{j} }{ V_{j} } \leq \frac{ v(q_{1}) + v(m_{0}) + v(m_{1}) + v(m_{2}) }{ v(q_{1}) + v(m_{1}) + v(m_{2}) }
			\leq R
	\]
	%
	
	%
	%
	%
	Case~3.1
	Case~3.1
	${CP}$
	\[	
		V_{j} = v(q_{1}) + v(m_{0}) + v(m_{1}) + v(m_{2})
	\]
	%
	%
	$P$
	$P(t, t+4, t+4) = \{ m_{0}, m_{1}, m_{2}, m_{3} \}$
	%
	%
	\[
		V'_{j} \leq V(t, t+4, t+4) = v(m_{0}) + v(m_{1}) + v(m_{2}) + v(m_{3})
	\]
	%
	%
	Case~3.1
	\[
		\frac{ V'_{j} }{ V_{j} } 
			\leq \frac{ v(m_{0}) + v(m_{1}) + v(m_{2}) + v(m_{3}) }{ v(q_{1}) + v(m_{0}) + v(m_{1}) + v(m_{2}) }
			\leq R
	\]
	%
	
	%
	Case~3.2.1
	Case~3.2.2
	Case~3.2.2
	${CP}$
	$m_{2}$
	\[
		V_{j} = v(q_{1}) + v(m_{0}) + v(m_{2})
	\]
	%
	%
	\[
		V'_{j} = v(m_{0}) + v(m_{1}) + v(m_{2})
	\]
	%
	%
	\begin{eqnarray*}
		\frac{ V'_{j} }{ V_{j} } 
			&\leq& 
				\frac{ v(m_{0}) + v(m_{1}) + v(m_{2}) }{ v(q_{1}) + v(m_{1}) + v(m_{2}) }\\
			&<& 
				\frac{ v(m_{0}) + v(m_{1}) + (R-1)(v(m_{0}) + v(m_{1})) + R v(q_{1}) }
					{ v(q_{1}) + v(m_{1}) + (R-1)(v(m_{0}) + v(m_{1})) + R v(q_{1}) } 
						\mbox{\hspace*{10mm} (
			&=& 
				\frac{ R v(m_{0}) + R v(m_{1}) + R v(q_{1}) }
					{ R v(m_{0}) + (R-1)v(m_{1}) + (R+1) v(q_{1}) }\\
			&<& 
				\frac{ R v(m_{0}) + R v(m_{1})) + R(R-1)(v(m_{0}) + v(m_{1})) }
					{ R v(m_{0}) + (R-1)v(m_{1}) + (R+1)(R-1)(v(m_{0}) + v(m_{1})) }
						\mbox{\hspace*{10mm} (
			&=& 
				\frac{ R^2 v(m_{0}) + R^2 v(m_{1}) }
					{ (R^2 + R - 1) v(m_{0})) + (R^2 + R - 2) v(m_{1}) } \\
			&<& \frac{ R^3 }{ 2 R^3 + R^2 - 4 R + 1 } 
			 		\mbox{\hspace*{10mm} (
			&<& 1.23 
			< R
	\end{eqnarray*}
	%
	
	%
	Case~3.2.3
	Case~3.2.3
	${CP}$
	\[
		V_{j} = v(q_{1}) + v(m_{0}) + v(m_{2}) + v(m_{3})
	\]
	%
	%
	$P$
	$P(t, t+3, t+3) = \{ m_{0}, m_{1}, m_{2}, m_{3} \}$
	Case~3.2.3
	$q_{3} = q_{1}$
	$P(t, t+3, t+4) = \{ q_{3}, m_{0}, m_{1}, m_{2}, m_{3} \} = \{ q_{1}, m_{0}, m_{1}, m_{2}, m_{3} \}$
	%
	%
	$OPT$
	\[
		V'_{j} \leq V(t, t+3, t+3) = v(m_{0}) + v(m_{1}) + v(m_{2}) + v(m_{3})
	\]
	%
	%
	\begin{eqnarray*}
		\frac{ V'_{j} }{ V_{j} } 
			\leq \frac{ v(m_{0}) + v(m_{1}) + v(m_{2}) + v(m_{3}) }{ v(q_{1}) + v(m_{0}) + v(m_{2}) + v(m_{3}) }
			 <  \frac{ v(p) + v(m_{0}) + v(m_{1}) + v(m_{2}) + v(m_{3}) }{ v(q_{1}) + v(m_{0}) + v(m_{2}) + v(m_{3}) }
	\end{eqnarray*}
	%
	%
	$OPT$
	\[
		V'_{j} \leq V(t, t+3, t+4) = v(q_{1}) + v(m_{0}) + v(m_{1}) + v(m_{2}) + v(m_{3})
	\]
	%
	%
	\begin{eqnarray} \label{eq:L3.7}
		v(q_{1}) &+& v(m_{0}) + v(m_{1}) + v(m_{2}) + v(m_{3}) \nonumber \\
			&>& R (v(m_{0}) + v(m_{1}) + v(m_{2})) + (R+1) v(q_{1})
				\mbox{\hspace*{10mm} (
			&>& R^2 (v(m_{0}) + v(m_{1})) + (R^2+R+1) v(q_{1})
				\mbox{\hspace*{10mm} (
			&>& R^2 (v(m_{0}) + v(m_{1})) + (R^2+R+1)(R-1) (v(m_{0}) + v(m_{1}))
				\mbox{\hspace*{10mm} (
			&=& (R^3 + R^2 - 1) (v(m_{0}) + v(m_{1}))
	\end{eqnarray}
	\begin{eqnarray*}
		\frac{ V'_{j} }{ V_{j} } 
			&\leq& \frac{ v(q_{1}) + v(m_{0}) + v(m_{1}) + v(m_{2}) + v(m_{3}) }
				{ v(q_{1}) + v(m_{0}) + v(m_{2}) + v(m_{3}) } \\
			&<& \frac{ (R^3 + R^2 - 1) (v(m_{0}) + v(m_{1})) }
				{ (R^3 + R^2 - 1) (v(m_{0}) + v(m_{1})) - v(m_{1}) } 
					\mbox{\hspace*{10mm} (
			&<& \frac{ R^4 + R^3 - R }{ R^4 + 2 R^3 - 4 R + 1 } 
			 		\mbox{\hspace*{10mm} (
			&<& 1.23 
			< R
	\end{eqnarray*}
	%
	
	%
	Case~3.2.1
	Case~3.2.1
	${CP}$
	\[
		V_{j} = v(q_{1}) + v(m_{0}) + v(m_{2}) + v(m_{3})
	\]
	%
	$P(t, t+4, t+4) =  \{ m_{0}, m_{1}, m_{2}, m_{3} \}$
	\[
		V'_{j} \leq V(t, t+4, t+4) = v(m_{0}) + v(m_{1}) + v(m_{2}) + v(m_{3})
	\]
	%
	%
	\begin{eqnarray*}
		\frac{ V'_{j} }{ V_{j} } 
			&\leq& 
				\frac{ v(m_{0}) + v(m_{1}) + v(m_{2}) + v(m_{3}) }{ v(q_{1}) + v(m_{0}) + v(m_{2}) + v(m_{3}) } \\
			&<& \frac{ v(q_{1}) + v(m_{0}) + v(m_{1}) + v(m_{2}) + v(m_{3}) }
				{ v(q_{1}) + v(m_{0}) + v(m_{2}) + v(m_{3}) } 
			< R
	\end{eqnarray*}
	%
	
	%
	\fi
	\ifnum \count11 > 0
	%
	%
	Suppose that ${CP}$ executes Case~1 at a time $t$ and 
	$t$ is contained in $T_{j}$. 
	Note that 
	if ${CP}$ executes Case~1.2.3.4 at $t$, 
	then ${CP}$ executes Case~2 at time $t+1$. 
	Moreover, 
	if ${CP}$ executes Case~2.2.2.3 at $t+1$, 
	then ${CP}$ executes Case~3 at time $t+2$. 
	For ease of presentation, 
	$m_{i}$ and $q_{i}$ denote $m_{i}(t)$ and $q_{i}(t)$, respectively. 
	Before proceeding to the proof, 
	we give some inequalities used often later. 
	When ${CP}$ executes Case~1.2.3.4 at $t$, 
	by the condition of Case~1.2.3.4, 
	\begin{equation} \label{eq:L3.1}
		v(m_{0}) < v(m_{1}) 
	\end{equation}
	and 
	\begin{equation} \label{eq:L3.2}
		v(q_{1}) > (R-1)(v(m_{0}) + v(m_{1})). 
	\end{equation}
	When ${CP}$ executes Case~2.2 at $t+1$, 
	\begin{equation} \label{eq:L3.3}
		v(m_{2}) > (R-1)(v(m_{0}) + v(m_{1})) + R v(q_{1}). 
	\end{equation}
	When ${CP}$ executes Case~3.2 at $t+2$, 
	\begin{equation} \label{eq:L3.5}
		v(m_{3}) > (R-1)(v(m_{0}) + v(m_{1}) + v(m_{2})) + R v(q_{1}). 
	\end{equation}
	When ${CP}$ executes Case~2.2.2.3 at $t+1$, 
	\begin{eqnarray*} 
		&& v(m_{0}) 
			> (R - 1)(v(q_{2}) + v(m_{1}) + v(m_{2})) \\
			&=& (R - 1)(v(q_{1}) + v(m_{1}) + v(m_{2})) 
				\mbox{\hspace*{10mm} (by the condition of Case~2.2.2.3)}\\
			&>& (R - 1)(v(q_{1}) + v(m_{1})) + (R-1)( (R-1)(v(m_{0}) + v(m_{1})) + R v(q_{1}) ) 
				\mbox{\hspace*{5mm} (by Eq.~(\ref{eq:L3.3}))}\\
			&=& (R^2 - 1) v(q_{1}) + (R^2 - R) v(m_{1}) + (R^2 - 2R + 1) v(m_{0}) \\
			&>& (R^2 - 1)(R - 1)(v(m_{0} + v(m_{1})) + (R^2 - R) v(m_{1}) + (R^2 - 2R + 1) v(m_{0}) 
				\mbox{\hspace*{3mm} (by Eq.~(\ref{eq:L3.2}))}\\
			&=& (R^3 - 3R + 2) v(m_{0}) + (R^3 - 2R + 1) v(m_{1}). 
	\end{eqnarray*}
	By rearranging this inequality, 
	we have 
	\begin{equation} \label{eq:L3.6}
		v(m_{0}) > \frac{ R^3 - 2R + 1 }{ -R^3 + 3 R - 1 } v(m_{1}). 
	\end{equation}
	%
	
	%
	Now we discuss the profit ratio for the execution of each case. 
	When ${CP}$ executes Case~1.1, 
	${CP}$ transmits $m_{0}$ at $t$ and 
	$V_{j} = v(m_{0}) = V(t, t, t)$. 
	On the other hand, 
	$V'_{j} \leq V(t, t, t)$ by Lemma~\ref{LMA:L1}. 
	Thus, 
	$V'_{j} / V_{j} \leq 1$. 
	%
	
	%
	When Case~1.2.1 or 1.2.2 is executed, 
	${CP}$ transmits both $m_{0}$ and $m_{1}$ and 
	$V_{j} = v(m_{0}) + v(m_{1}) = V(t, t+1, t+1)$. 
	By Lemma~\ref{LMA:L1}, 
	$V'_{j} \leq V(t, t+1, t+1)$. 
	Thus, 
	$V'_{j} / V_{j} \leq 1$. 
	%
	
	%
	We consider Case~1.2.3.1. 
	${CP}$ transmits $q_{1}$ and $m_{0}$ at times $t$ and $t+1$, respectively, by definition. 
	Thus, 
	\[	
		V_{j} = v(q_{1}) + v(m_{0}). 
	\]
	Since $P(t, t+1, t+1) = \{ m_{0}, m_{1} \}$ by the definition of $P$, 
	by Lemma~\ref{LMA:L1}, 
	\[
		V'_{j} \leq V(t, t+1, t+1) = v(m_{0}) + v(m_{1}). 
	\]
	Hence, 
	\[
		\frac{ V'_{j} }{ V_{j} } \leq \frac{ v(m_{0}) + v(m_{1}) }{ v(q_{1}) + v(m_{0}) }
			\leq \frac{ v(m_{1}) + v(m_{1}) }{ \alpha v(m_{1}) + v(m_{1}) }
			= \frac{ 2 }{ \alpha + 1 } 
			= R, 
	\]
	in which 
	the second inequality follows from 
	$v(q_{1}) \geq \alpha v(m_{1})$ 
	and 
	$v(m_{0}) \geq v(m_{1})$, 
	which is the condition of Case~1.2.3.1, 
	and 
	the last equality follows from the definitions of $\alpha$ and $R$. 
	%
	
	%
	Let us consider Case~1.2.3.2. 
	Since ${CP}$ transmits $m_{0}$ and $m_{1}$ at $t$ and $t+1$, respectively, by definition, 
	\[
		V_{j} = v(m_{0}) + v(m_{1}). 
	\]
	$P(t, t+1, t+1) = \{ m_{0}, m_{1} \}$
	and 
	$P(t, t+1, t+2) = \{ q_{1}, m_{0}, m_{1} \}$
	by definition. 
	If $OPT$ does not transmit $m_{1}$ at $t+2$, 
	we have using Lemma~\ref{LMA:L1}, 
	\[
		V'_{j} \leq V(t, t+1, t+1) = v(m_{0}) + v(m_{1}). 
	\]
	Thus, 
	$V'_{j} / V_{j} \leq 1$. 
	If $OPT$ transmits $m_{1}$ at $t+2$, 
	we have using Lemma~\ref{LMA:L2}, 
	\[
		V'_{j} \leq V(t, t+1, t+2) = v(q_{1}) + v(m_{0}) + v(m_{1}). 
	\]
	By these inequalities, 
	\[
		\frac{ V'_{j} }{ V_{j} } \leq \frac{ v(q_{1}) + v(m_{0}) + v(m_{1}) }{ v(m_{0}) + v(m_{1}) }
			< \frac{ \alpha v(m_{0}) + v(m_{0}) + v(m_{0}) }{ v(m_{0}) + v(m_{0}) }
			= \frac{ \alpha + 2 }{ 2 } 
			= R, 
	\]
	in which 
	the second inequality follows from $v(q_{1}) < \alpha v(m_{1})$ 
	and 
	$v(m_{0}) \geq v(m_{1})$, 
	which is the execution condition of Case~1.2.3.2, 
	and 
	the last equality is immediately from the definitions of $\alpha$ and $R$. 
	%
	
	%
	We consider Case~1.2.3.3. 
	Since ${CP}$ transmits $m_{0}$ and $m_{1}$ at $t$ and $t+1$, respectively, 
	\[	
		V_{j} = v(m_{0}) + v(m_{1}). 
	\]
	In a similar way to the proof of Case~1.2.3.2, 
	if $OPT$ does not transmit $m_{1}$ at $t+2$, 
	it follows from Lemma~\ref{LMA:L1} that 
	\[
		V'_{j} \leq V(t, t+1, t+1) = v(m_{0}) + v(m_{1}). 
	\]
	Thus, 
	$V'_{j} / V_{j} \leq 1$. 
	If $OPT$ transmits $m_{1}$ at time $t+2$, 
	we have by Lemma~\ref{LMA:L2}, 
	\[
		V'_{j} \leq V(t, t+1, t+2) = v(q_{1}) + v(m_{0}) + v(m_{1}). 
	\]
	By the condition of Case~1.2.3.3, 
	\[
		\frac{ V'_{j} }{ V_{j} } \leq \frac{ v(q_{1}) + v(m_{0}) + v(m_{1}) }{ v(m_{0}) + v(m_{1}) }
			\leq R. 
	\]
	%
	
	%
	For the rest of the proof, 
	suppose that ${CP}$ executes Case~1.2.3.4 at $t$ and next 
	executes Case~2 at $t+1$. 
	Hence, 
	${CP}$ transmits $q_{1}$ at $t$. 
	First, 
	we consider the case in which Case~2.1 is executed at $t+1$. 
	By the definition of Case~2.1, 
	${CP}$ transmits $m_{0}$ and $m_{1}$ at $t+1$ and $t+2$, respectively, 
	and thus,  
	\[
		V_{j} = v(q_{1}) + v(m_{0}) + v(m_{1}). 
	\]
	On the other hand, 
	since $P(t, t+2, t+2) =  \{ m_{0}, m_{1}, m_{2} \}$ by definition, 
	it follows from Lemma~\ref{LMA:L1} that 
	\[
		V'_{j} \leq V(t, t+2, t+2) = v(m_{0}) + v(m_{1}) + v(m_{2}). 
	\]
	Thus, 
	\[
		\frac{ V'_{j} }{ V_{j} } 
			\leq \frac{ v(m_{0}) + v(m_{1}) + v(m_{2}) }{ v(q_{1}) + v(m_{0}) + v(m_{1}) }
			\leq R, 
	\]
	which follows from the condition of Case~2.1. 
	%
	
	%
	We discuss Case~2.2.1. 
	Since ${CP}$ transmits $m_{1}$ and $m_{2}$ at $t+1$ and $t+2$, respectively, by definition, 
	\[
		V_{j} = v(q_{1}) + v(m_{1}) + v(m_{2}). 
	\]
	Since $P(t, t+2, t+2) =  \{ m_{0}, m_{1}, m_{2} \}$, 
	we have using Lemma~\ref{LMA:L1}, 
	\[
		V'_{j} \leq V(t, t+1, t+1) = v(m_{0}) + v(m_{1}) + v(m_{2}). 
	\]
	Hence, 
	\begin{eqnarray*}
		\frac{ V'_{j} }{ V_{j} } 
			&\leq& \frac{ v(m_{0}) + v(m_{1}) + v(m_{2}) }{ v(q_{1}) + v(m_{1}) + v(m_{2}) }\\
			&<&  \frac{ v(m_{0}) + v(m_{1}) + (R-1)(v(m_{0}) + v(m_{1})) + R v(q_{1}) }{ v(q_{1}) + v(m_{1}) + (R-1)(v(m_{0}) + v(m_{1})) + R v(q_{1}) } \mbox{\hspace*{10mm} (by Eq.~(\ref{eq:L3.3}))}\\
			&=&  \frac{ R (v(m_{0}) + v(m_{1})) + R v(q_{1}) }{ R(v(q_{1}) + v(m_{1})) - v(m_{0}) + (R+1) v(q_{1}) }\\
			&<& \frac{ R (v(m_{0}) + v(m_{1})) + R(R-1)(v(m_{0}) + v(m_{1})) }{ R(v(q_{1}) + v(m_{1})) - v(m_{0}) + (R+1)(R-1)(v(m_{0}) + v(m_{1})) } \mbox{\hspace*{10mm} (by Eq.~(\ref{eq:L3.2}))}\\
			&=& \frac{ R^2 (v(m_{0}) + v(m_{1})) }{ (R^2 + R - 1)(v(q_{1}) + v(m_{1})) - v(m_{0}) } 
			 <  \frac{ 2 R^2 v(m_{0}) }{ 2(R^2 + R - 1) v(m_{0}) - v(m_{0}) } \mbox{\hspace*{5mm} (by Eq.~(\ref{eq:L3.1}))}\\
			&=& \frac{ 2 R^2 }{ 2 R^2 + 2 R - 3 } < \frac{ R^2 }{ R } = R. 
	\end{eqnarray*}
	%
	
	%
	In Case~2.2.2.1, 
	${CP}$ transmits $m_{1}$ at $t+1$ and hence, 
	\[
		V_{j} = v(q_{1}) + v(m_{1}). 
	\]
	By Lemma~\ref{LMA:L5}, 
	\[
		V'_{j} = v(m_{0}) + v(m_{1}). 
	\]
	Thus, 
	\begin{eqnarray*}
		\frac{ V'_{j} }{ V_{j} } 
			&\leq& \frac{ v(m_{0}) + v(m_{1}) }{ v(q_{1}) + v(m_{1}) } \\
			&<& \frac{ v(m_{0}) + v(m_{1}) }{ (R-1)(v(m_{0}) + v(m_{1})) + v(m_{1}) } \mbox{\hspace*{10mm} (by Eq.~(\ref{eq:L3.2}))}\\
			&<& \frac{ v(m_{0}) + v(m_{0}) }{ (R-1)(v(m_{0}) + v(m_{0})) + v(m_{0}) } \mbox{\hspace*{10mm} (by Eq.~(\ref{eq:L3.1}))}\\
			&=& \frac{ 2 }{ R } = \sqrt{5} - 1 < R. 
	\end{eqnarray*}
	%
	
	%
	In Case~2.2.2.2, 
	${CP}$ transmits $m_{1}$ and $m_{2}$ at $t+1$ and $t+2$, respectively, 
	\[	
		V_{j} = v(q_{1}) + v(m_{1}) + v(m_{2}). 
	\]
	By the definition of $P$, 
	$P(t, t+2, t+2) = \{ m_{0}, m_{1}, m_{2} \}$. 
	Since 
	$q_{2} = q_{1}$ 
	by the condition of Case~2.2.2.2,
	$P(t, t+2, t+3) = \{ q_{2}, m_{0}, m_{1}, m_{2} \} = \{ q_{1}, m_{0}, m_{1}, m_{2} \}$. 
	If $OPT$ does not transmit $m_{2}$ at $t+3$, 
	by Lemma~\ref{LMA:L1}, 
	\[
		V'_{j} \leq V(t, t+2, t+2) = v(m_{0}) + v(m_{1}) + v(m_{2}). 
	\]
	Thus, 
	\[
		\frac{ V'_{j} }{ V_{j} } 
			\leq \frac{ v(m_{0}) + v(m_{1}) + v(m_{2})}{ v(q_{1}) + v(m_{1}) + v(m_{2}) }
			 <   \frac{ v(q_{1}) + v(m_{0}) + v(m_{1}) + v(m_{2})}{ v(q_{1}) + v(m_{1}) + v(m_{2}) }. 
	\]
	If $OPT$ transmits $m_{2}$ at $t+3$, 
	by Lemma~\ref{LMA:L2}, 
	\[
		V'_{j} \leq V(t, t+2, t+3) = v(q_{1}) + v(m_{0}) + v(m_{1}) + v(m_{2}). 
	\]
	Therefore, 
	\[
		\frac{ V'_{j} }{ V_{j} } \leq \frac{ v(q_{1}) + v(m_{0}) + v(m_{1}) + v(m_{2}) }{ v(q_{1}) + v(m_{1}) + v(m_{2}) }
			\leq R, 
	\]
	which is immediately from the condition of Case~2.2.2.2. 
	%
	
	%
	In the following, 
	suppose that ${CP}$ executes Cases~2.2.2.3 and 3 at $t+1$ and $t+2$, respectively, 
	which indicates that 
	${CP}$ transmits $m_{0}$ at $t+1$. 
	%
	%
	Let us consider the case in which Case~3.1 is executed at $t+2$. 
	Since ${CP}$ transmits $m_{1}$ and $m_{2}$ at $t+2$ and $t+3$, respectively, 
	\[
		V_{j} = v(q_{1}) + v(m_{0}) + v(m_{1}) + v(m_{2}). 
	\]
	$P(t, t+4, t+4) = \{ m_{0}, m_{1}, m_{2}, m_{3} \}$ and thus 
	we have using Lemma~\ref{LMA:L1}, 
	\[
		V'_{j} \leq V(t, t+4, t+4) = v(m_{0}) + v(m_{1}) + v(m_{2}) + v(m_{3}). 
	\]
	Thus, 
	\[
		\frac{ V'_{j} }{ V_{j} } 
			\leq \frac{ v(m_{0}) + v(m_{1}) + v(m_{2}) + v(m_{3}) }{ v(q_{1}) + v(m_{0}) + v(m_{1}) + v(m_{2}) }
			\leq R, 
	\]
	which is immediately from the condition of Case~3.1. 
	%
	
	%
	We discuss Case~3.2.1 at the end of this proof 
	and next consider Case~3.2.2. 
	Since ${CP}$ transmits $m_{2}$ at $t+2$, 
	\[
		V_{j} = v(q_{1}) + v(m_{0}) + v(m_{2}). 
	\]
	Moreover, 
	by Lemma~\ref{LMA:L6}, 
	\[
		V'_{j} = v(m_{0}) + v(m_{1}) + v(m_{2}). 
	\]
	Hence, 
	\begin{eqnarray*}
		\frac{ V'_{j} }{ V_{j} } 
			&\leq& 
				\frac{ v(m_{0}) + v(m_{1}) + v(m_{2}) }{ v(q_{1}) + v(m_{1}) + v(m_{2}) }\\
			&<& 
				\frac{ v(m_{0}) + v(m_{1}) + (R-1)(v(m_{0}) + v(m_{1})) + R v(q_{1}) }
					{ v(q_{1}) + v(m_{1}) + (R-1)(v(m_{0}) + v(m_{1})) + R v(q_{1}) } 
						\mbox{\hspace*{10mm} (by Eq.~(\ref{eq:L3.3}))}\\
			&=& 
				\frac{ R v(m_{0}) + R v(m_{1}) + R v(q_{1}) }
					{ R v(m_{0}) + (R-1)v(m_{1}) + (R+1) v(q_{1}) }\\
			&<& 
				\frac{ R v(m_{0}) + R v(m_{1})) + R(R-1)(v(m_{0}) + v(m_{1})) }
					{ R v(m_{0}) + (R-1)v(m_{1}) + (R+1)(R-1)(v(m_{0}) + v(m_{1})) }
						\mbox{\hspace*{10mm} (by Eq.~(\ref{eq:L3.2}))}\\
			&=& 
				\frac{ R^2 v(m_{0}) + R^2 v(m_{1}) }
					{ (R^2 + R - 1) v(m_{0})) + (R^2 + R - 2) v(m_{1}) } \\
			&<& \frac{ R^3 }{ 2 R^3 + R^2 - 4 R + 1 } 
			 		\mbox{\hspace*{10mm} (by Eq.~(\ref{eq:L3.6}))}\\
			&<& 1.23 
			< R. 
	\end{eqnarray*}
	%
	
	%
	In Case~3.2.3, 
	${CP}$ transmits $m_{2}$ and $m_{3}$ at $t+2$ and $t+3$, respectively, 
	and thus, 
	\[
		V_{j} = v(q_{1}) + v(m_{0}) + v(m_{2}) + v(m_{3}). 
	\]
	On the other hand, 
	$P(t, t+3, t+3) = \{ m_{0}, m_{1}, m_{2}, m_{3} \}$. 
	Since $q_{3} = q_{1}$ by the condition of Case~3.2.3, 
	$P(t, t+3, t+4) = \{ q_{3}, m_{0}, m_{1}, m_{2}, m_{3} \} = \{ q_{1}, m_{0}, m_{1}, m_{2}, m_{3} \}$. 
	If $OPT$ does not transmit $m_{3}$ at $t+4$, 
	by Lemma~\ref{LMA:L1}, 
	\[
		V'_{j} \leq V(t, t+3, t+3) = v(m_{0}) + v(m_{1}) + v(m_{2}) + v(m_{3}). 
	\]
	Hence, 
	\begin{eqnarray*}
		\frac{ V'_{j} }{ V_{j} } 
			\leq \frac{ v(m_{0}) + v(m_{1}) + v(m_{2}) + v(m_{3}) }{ v(q_{1}) + v(m_{0}) + v(m_{2}) + v(m_{3}) }
			 <  \frac{ v(p) + v(m_{0}) + v(m_{1}) + v(m_{2}) + v(m_{3}) }{ v(q_{1}) + v(m_{0}) + v(m_{2}) + v(m_{3}) }. 
	\end{eqnarray*}
	If $OPT$ transmits $m_{3}$ at $t+4$, 
	by Lemma~\ref{LMA:L2}, 
	\[
		V'_{j} \leq V(t, t+3, t+4) = v(q_{1}) + v(m_{0}) + v(m_{1}) + v(m_{2}) + v(m_{3}). 
	\]
	Then, 
	\begin{eqnarray} \label{eq:L3.7}
		v(q_{1}) &+& v(m_{0}) + v(m_{1}) + v(m_{2}) + v(m_{3}) \nonumber \\
			&>& R (v(m_{0}) + v(m_{1}) + v(m_{2})) + (R+1) v(q_{1})
				\mbox{\hspace*{10mm} (by Eq.~(\ref{eq:L3.5}))} \nonumber \\
			&>& R^2 (v(m_{0}) + v(m_{1})) + (R^2+R+1) v(q_{1})
				\mbox{\hspace*{10mm} (by Eq.~(\ref{eq:L3.3}))} \nonumber \\
			&>& R^2 (v(m_{0}) + v(m_{1})) + (R^2+R+1)(R-1) (v(m_{0}) + v(m_{1}))
				\mbox{\hspace*{10mm} (by Eq.~(\ref{eq:L3.2}))} \nonumber \\
			&=& (R^3 + R^2 - 1) (v(m_{0}) + v(m_{1}))
	\end{eqnarray}
	Hence, we have 
	\begin{eqnarray*}
		\frac{ V'_{j} }{ V_{j} } 
			&\leq& \frac{ v(q_{1}) + v(m_{0}) + v(m_{1}) + v(m_{2}) + v(m_{3}) }
				{ v(q_{1}) + v(m_{0}) + v(m_{2}) + v(m_{3}) } \\
			&<& \frac{ (R^3 + R^2 - 1) (v(m_{0}) + v(m_{1})) }
				{ (R^3 + R^2 - 1) (v(m_{0}) + v(m_{1})) - v(m_{1}) } 
					\mbox{\hspace*{10mm} (by Eq.~(\ref{eq:L3.7}))}\\
			&<& \frac{ R^4 + R^3 - R }{ R^4 + 2 R^3 - 4 R + 1 } 
			 		\mbox{\hspace*{10mm} (by Eq.~(\ref{eq:L3.6}))}\\
			&<& 1.23 
			< R. 
	\end{eqnarray*}
	%
	
	%
	Finally we discuss Case~3.2.1. 
	Since ${CP}$ transmits $m_{2}$ and $m_{3}$ at $t+2$ and $t+3$, respectively, 
	\[
		V_{j} = v(q_{1}) + v(m_{0}) + v(m_{2}) + v(m_{3}). 
	\]
	Since $P(t, t+4, t+4) =  \{ m_{0}, m_{1}, m_{2}, m_{3} \}$, 
	\[
		V'_{j} \leq V(t, t+4, t+4) = v(m_{0}) + v(m_{1}) + v(m_{2}) + v(m_{3}) 
	\]
	by Lemma~\ref{LMA:L1}. 
	In the same way as the proof of Case~3.2.2, 
	\begin{eqnarray*}
		\frac{ V'_{j} }{ V_{j} } 
			&\leq& 
				\frac{ v(m_{0}) + v(m_{1}) + v(m_{2}) + v(m_{3}) }{ v(q_{1}) + v(m_{0}) + v(m_{2}) + v(m_{3}) } \\
			&<& \frac{ v(q_{1}) + v(m_{0}) + v(m_{1}) + v(m_{2}) + v(m_{3}) }
				{ v(q_{1}) + v(m_{0}) + v(m_{2}) + v(m_{3}) } 
			< R. 
	\end{eqnarray*}
	\fi
\end{proof}
\ifnum \count10 > 0
%
%

%

%
\fi
\ifnum \count11 > 0
%
%

%
\fi
%


%
%
\ifnum \count10 > 0
%

%

%
\fi
\ifnum \count11 > 0
%
%
%
\fi
%



%
\ifnum \count10 > 0
%
%

%

%
\fi
\ifnum \count11 > 0
%
%

%
\fi

\end{document}